\documentclass[12pt, onecolumn]{IEEEtran}

\usepackage{graphicx}
\usepackage{color}
\usepackage{amssymb}
\usepackage{mathrsfs}
\usepackage{epsfig}
\usepackage{amsmath}

\newtheorem{Theorem}{Theorem}
\newtheorem{Corollary}{Corollary}
\newtheorem{Lemma}{Lemma}

\newcommand{\bs}{\backslash}

\begin{document}

\bibliographystyle{IEEEtran} 

\title{On the $k$-pairs problem}  
\author{Ali Al-Bashabsheh and Abbas Yongacoglu           
              \\School of Information Technology and Engineering
              \\University of Ottawa, Ottawa, Canada
              \\\{aalba059, yongacog\}@site.uottawa.ca
}

\date{}
\maketitle

\begin{abstract} We consider network coding rates for directed and undirected $k$-pairs networks. For directed networks, meagerness is known to be an upper bound on network coding rates. We show that network coding rate can be $\Theta(|V|)$ multiplicative factor smaller than meagerness. For the undirected case, we show some progress in the direction of the $k$-pairs conjecture.
\end{abstract}

\baselineskip 20pt

\section{Introduction}

It is known that the min-cut is a necessary and sufficient condition for achievable throughputs in multicast networks \cite{Yeung2000:1}. 
In general networks, the min-cut is not a sufficient condition and there is no answer yet to what rates are achievable in such networks. In this work we consider the $k$-\emph{pairs problem} which is also referred to as the \emph{multiple unicast sessions} problem. For directed networks, it is known that network coding may provide higher rates than routing. On the other hand, for undirected $k$-pairs networks, Li and Li conjectured that network coding can not provide higher rates than fractional routing \cite{Li2004:2}. This conjecture has been verified to be true for few networks \cite{Li2004:2} \cite{Yeung2005:1} \cite{Harvey2006:1} \cite{Savari2006:1}.

\section{Definitions and Problem Formulation}

A directed graph $G(V,E)$ is specified by a set of nodes $V$ and a set of directed edges $E$ (an incidence function is not necessary since in this work all graphs are assumed to be simple graphs). For any edge $e = (u,v)$ we write $\mbox{head}(e) = v$ and $\mbox{tail}(e) = u$. For a node $v \in V$ we denote by In$(v) = \{e \in E: \mbox{head}(e) = v\}$ the set of all edges going into $v$ and by Out$(v) = \{e \in E: \mbox{tail}(e) = v\}$ the set of all edges departing from $v$. Moreover, for any $U \subseteq V$ we use In$(U) = \{e \in E: \mbox{head}(e) \in U, \ \mbox{tail}(e) \notin U\}$ to indicate the set of all edges entering $U$. Similarly, Out$(U) = \{e \in E: \mbox{tail}(e) \in U, \ \mbox{head}(e) \notin U\}$ denotes the set of edges outgoing from $U$. 

Similar to directed graphs, an undirected graph $G(V,E)$ is specified by two sets $V$ and $E$ where edges do not have a prespecified direction and can provide a bidirectional transportation of information. At some places, an undirected edge $e$ between nodes $u$ and $v$ might be replaced with two directed edges $(u,v)$ and $(v,u)$ whose capacities sum to the capacity of $e$. For notational ease, we might drop the parenthesis and use $uv$ and $vu$ to denote the edges directed from $u$ to $v$ and $v$ to $u$, respectively, while preserving the notation $\{u,v\}$ for the undirected edge between $u$ and $v$. The set of all directed edges obtained from $E$ will be denoted $E_{d}$, i.e., $E_{d} = \{(u,v):\{u,v\} \in E\}$. In this work, all edges are assumed to have unit capacity.

A directed (undirected) $k$-pairs network consists of an underlying directed (undirected) graph, $G(V,E)$, and a set of $k$ source-sink pairs. A source-sink pair uniquely identifies a commodity to be communicated from the source to the sink. Let ${\cal I} = \{1,2,\ldots,k\}$ be the set of commodities, then for any $i \in {\cal I}$ we use $s(i) \in V$ and $t(i) \in V$ to denote the nodes which (respectively) generates and demands commodity $i$. We refer to $s(i)$ as the \emph{source node} and $t(i)$ as the \emph{sink node} of $i \in {\cal I}$ and always assume $s(i) \neq t(i)$. Note that a node $v \in V$ can be a source node or a sink node for more than one commodity. 
We denote by $S(v)$ the set of all commodities for which $v$ is a source node, i.e., $S(v) = \{i \in {\cal I}: s(i) = v\}$. Similarly, let $T(v)$ be the set of all commodities for which $v$ is a sink node, i.e., $T(v) = \{i \in {\cal I}: t(i) = v\}$. Also for any set of nodes $U \subseteq V$, let $S(U) = \{i \in {\cal I}: s(i) \in U \}$ be the set of all commodities whose sources are in $U$ and $T(U) =\{i \in {\cal I}: t(i) \in U \}$ be the set of all commodities whose sinks are in $U$.

Given an undirected network, ${\cal N}$, with an underlying graph $G(V,E)$ and a set of commodities ${\cal I}$. A set of edges $A \subseteq E$ is said to \emph{separate} commodity $i \in {\cal I}$ if every path from $s(i)$ to $t(i)$ contains at least one edge from $A$. Let $\cal{J}$ be the set of commodities separated by $A$, then \emph{sparsity} \cite{Farhad1990:1} \cite{Lehman2005:1} of $A$ is defined as ${\mathscr S}(A) = |A|/|{\cal J}|$. Moreover, the sparsity of the graph is defined as ${\mathscr S}_{G} = \min_{A \subseteq E} {\mathscr S}(A)$. 
 It is clear that sparsity is a bottleneck for the communication pairs (Indeed, some authors refer to ${\mathscr S}_{G}$ as the min-cut bound \cite{Farhad1990:1}). Thus, in undirected networks, sparsity is an upper bound on achievable rates with or without network coding. 
Another bound on routing rates can be defined in terms of the Wiener index. For any pair of nodes $u \neq v \in V$ let $d(u,v)$ be the number of edges in the shortest path between $u$ and $v$ in $G$. 
 The \emph{wiener index} \cite{Wiener1947:1} of a graph $G$ is defined as $D_{G} = \sum_{\{u,v\} \subseteq V}{d(u,v)}$ which is a commonly used quantity in chemical literature. We define the wiener index of the network as $D_{\cal N} = \sum_{i \in {\cal I}}{d(s(i),t(i))}$. Obviously, if their is a commodity between every pair of distinct nodes in the network, i.e. $|{\cal I}| = \binom{|V|}{2}$, then the wiener indices of the graph and the network are equal. The \emph{wiener bound} of the network $\cal N$ is defined as $W_{\cal N} = |E|/D_{\cal N}$. Clearly $W_{{\cal N}}$ is an upper bound on achievable routing rates in undirected $k$-pairs networks. This follows since in routing, if an edge is used to transpose a fraction of commodity $i$, then an equal fraction of the edge capacity is exclusively used by such commodity, i.e, an edge does not carry a combination of messages from different commodities.

For directed networks, sparsity is still an upper bound on routing rates but it is not an upper bound on network coding rates. 
 \emph{Meagerness} was introduced in \cite{Lehman2005:1} to bound network coding rates in directed networks. For any set $A \subseteq E$ of edges and a set of commodities ${\cal J} \subseteq {\cal I}$ we say $A$ \emph{isolates} ${\cal J}$ if every path from $s(i)$ to $t(j)$ $\forall i, j \in {\cal J}$ contains at least one edge from $A$. The meagerness of set $A$ is defined as 
\[{\cal M}(A) = \min_{{\cal J}: A \ \mbox{isolates} \ {\cal J}}  \frac{|A|}{|{\cal J}|}  \] 
and the meagerness of the network $\cal N$ is defined as ${\cal M}_{\cal N} = \min_{A \subseteq E} {\cal M}(A)$.

With each commodity $i \in {\cal I}$ we associate a R.V. $X_{i}$ which represents a message generated at $s(i)$ and to be correctly recovered at $t(i)$. For notational convenience we might use set subscript. More specifically, let $A$ be any set, then $X_{A} = \{X_{a}:a \in A\}$ (if $A$ is empty we set $X_{A}$ to be a constant). Also with each directed edge $e = uv$ we associate a R.V. $X_{uv}$ which is a deterministic function of $X_{S(u)}$ and $X_{\mbox{In}(u)}$. A sink node $t(i)$ must be able to recover its message using only the information available from $\mbox{In}(t(i))$ and $\{X_{j}:j \in S(t(i)) \}$. In other words, each sink recovers its message by computing a function of $X_{\mbox{In}(t(i))}$ and $X_{S(t(i))}$. The set of edge functions and sink functions defines a \emph{network code}. Such network code implies an achievable rate tuple $(r_{1}, \ldots, r_{k})$ where $r_{i} \leq H(X_{i})$ is the rate at which the $i$th commodity is communicated. Obviously, $H(X_{e})$ must not exceed the edge capacity $\forall e \in E$. An \emph{achievable symmetric rate} is the rate tuple $(r,r,\ldots,r)$ which can be uniquely identified with the scalar $r$. The \emph{network coding rate} is defined as the supremum of all achievable symmetric rates with network coding. 

The condition that $X_{uv}$ is a function of $X_{S(u)}$ and $X_{\mbox{In}(u)}$ is equivalent to $H(X_{\mbox{In}(u)},X_{S(u)},X_{uv})$ = $H(X_{\mbox{In}(u)},X_{S(u)})$ since $H(X_{uv}|X_{\mbox{In}(u)},X_{S(u)}) = 0$. From monotonicity of entropy, the previous equality can be written as $H(X_{\mbox{In}(u)},X_{S(u)},X_{uv}) \leq H(X_{\mbox{In}(u)},X_{S(u)})$. This has motivated the authors in \cite{Yeung2005:1} to define the \emph{input-output} inequality which states that for any $U \subseteq V$, $H(X_{\mbox{In}(U)},X_{S(U)},X_{\mbox{Out}(U)},X_{T(U)} ) \leq H(X_{\mbox{In}(U)},X_{S(U)})$. Finally, at some places we use the \emph{submodularity} of entropy which asserts that for any sets $A_{1}$ and $A_{2}$ we have $H(X_{A_{1}}) + H(X_{A_{2}}) \geq H(X_{A_{1} \cup A_{2}}) + H(X_{A_{1} \cap A_{2}})$.

\section{Directed Networks}

Meagerness was introduced in \cite{Lehman2005:1} to bound network coding rates in directed networks. In the same work, the authors provided a network referred to as the split butterfly to illustrate that the meagerness bound might not be tight. In this section we show that network coding rate can be $\Theta(|V|)$ multiplicative factor smaller than meagerness. This shows that for some networks which exhibit some topological asymmetries, meagerness may become too loose and terribly fails to tightly bound such networks' coding rates\footnote{Recently, it was brought to our attention (see acknowledgment) that a similar result was obtained in \cite{Harvey2005:1}. However, a different network topology was used in the proof. }.

Let ${\cal N}_{1}$ be a directed $k$-pairs network with a set of commodities ${\cal I} = \{1,2,\ldots,k\}$. The nodes of the underlying graph consist of $k$ source nodes $s(1), \ldots, s(k)$, two intermediate nodes $u,v$, and $k$ sink nodes $t(1),\ldots, t(k)$. The set of edges can be described as follows: There is an edge from every source node to the intermediate node $u$ and there is an edge from $v$ to every sink node. A single edge connects $u$ to $v$. Finally, every sink node $t(i)$ has an incoming edge from $s(j)$ $\forall i < j$. Fig.\ref{fig:Meagerness} shows network ${\cal N}_{1}$. 

\begin{figure}[htbp]
		\centering \setlength{\unitlength}{0.38cm}
		\begin{small}
		\begin{picture}(18,12.5)(-7.8,1.5)
		\put(-1,1.5){\includegraphics[width = 5cm,height = 4.2cm]{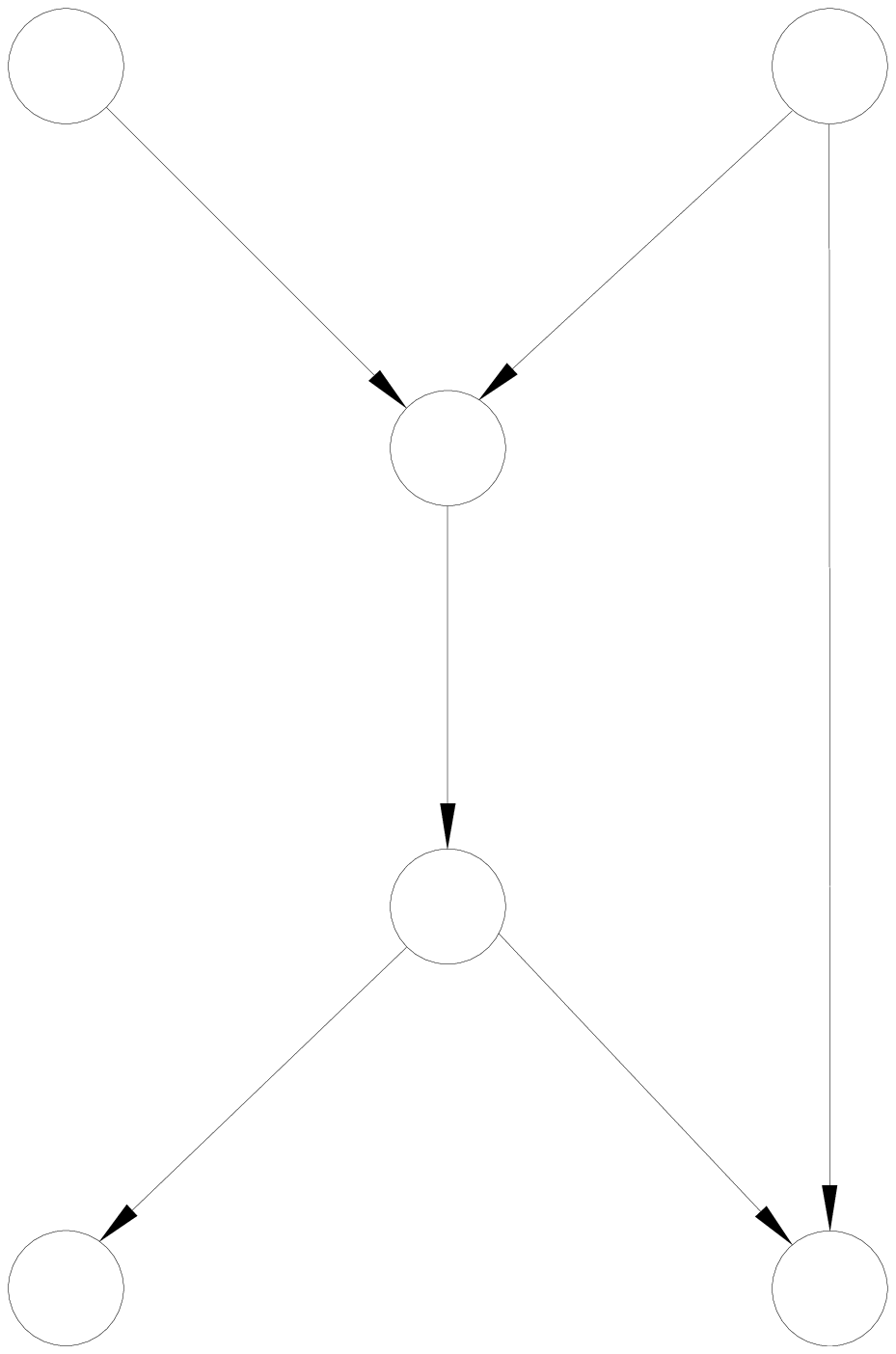} }
		\put(-1.2,12.7){$s(1)$} \put(4.2,12.7){$s(2)$} \put(1.3,8.6){$u$} \put(1.3,5){$v$} 
		\put(-1.2,.8){$t(2)$} \put(4.2,.8){$t(1)$} \put(2,-.4){(b)}
		\end{picture}
	  
	  \begin{picture}(18,0)(-8,.5)
		\put(6,1.5){ \includegraphics[width = 5cm,height = 4.2cm]{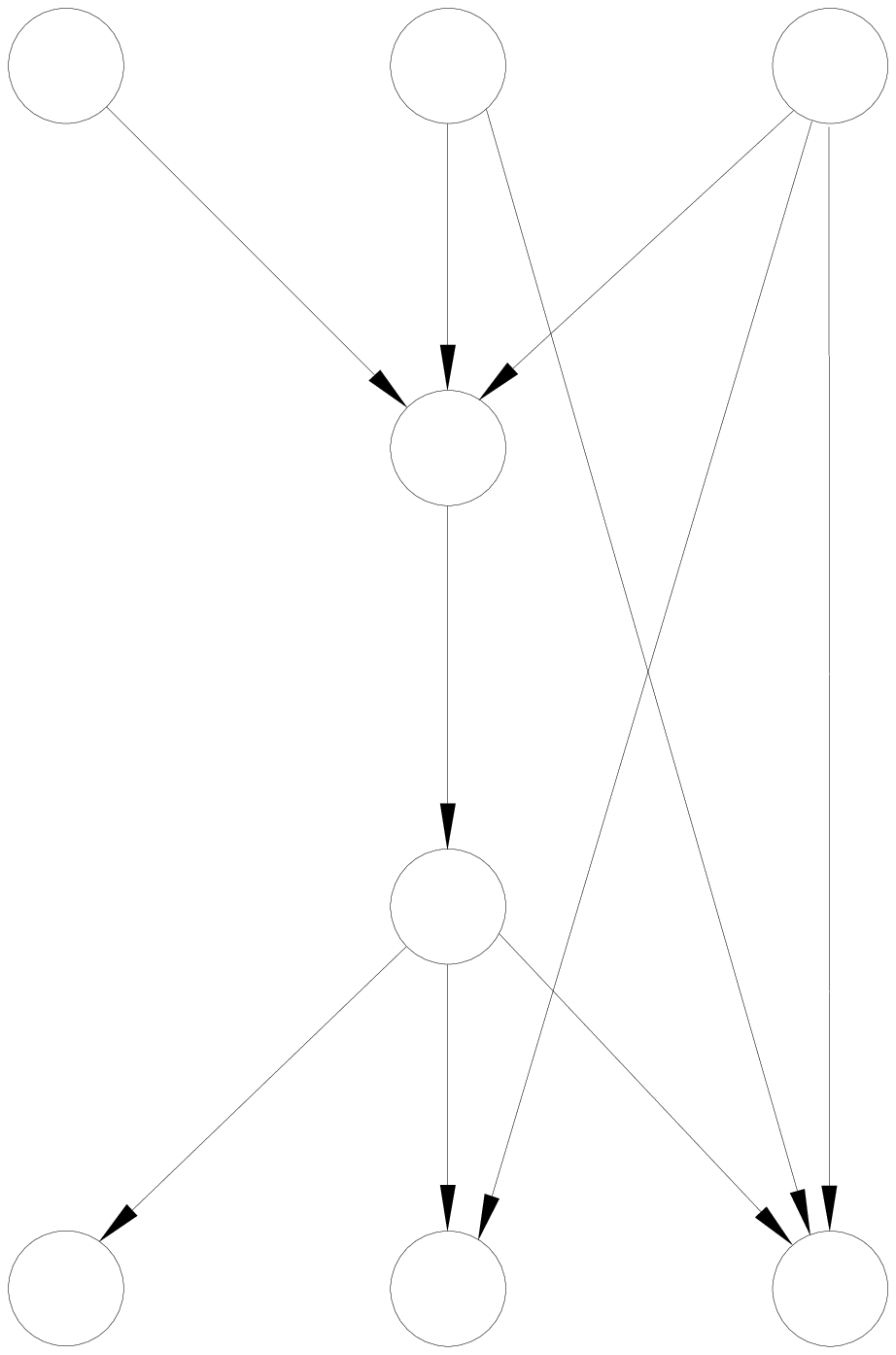} }
		\put(6.2,12.6){$s(1)$} \put(8.7,12.6){$s(2)$} \put(11.2,12.6){$s(3)$} \put(8.5,8.6){$u$} \put(8.5,5){$v$} 
		\put(6.2,.8){$t(3)$} \put(8.7,.8){$t(2)$} \put(11.2,.8){$t(1)$} \put(9.2,-.4){(c)}
		\end{picture}
		
		\begin{picture}(18,0)(15.5,-1.9)
		\put(13.4,0){ \includegraphics[width = 5cm,height = 4.2cm]{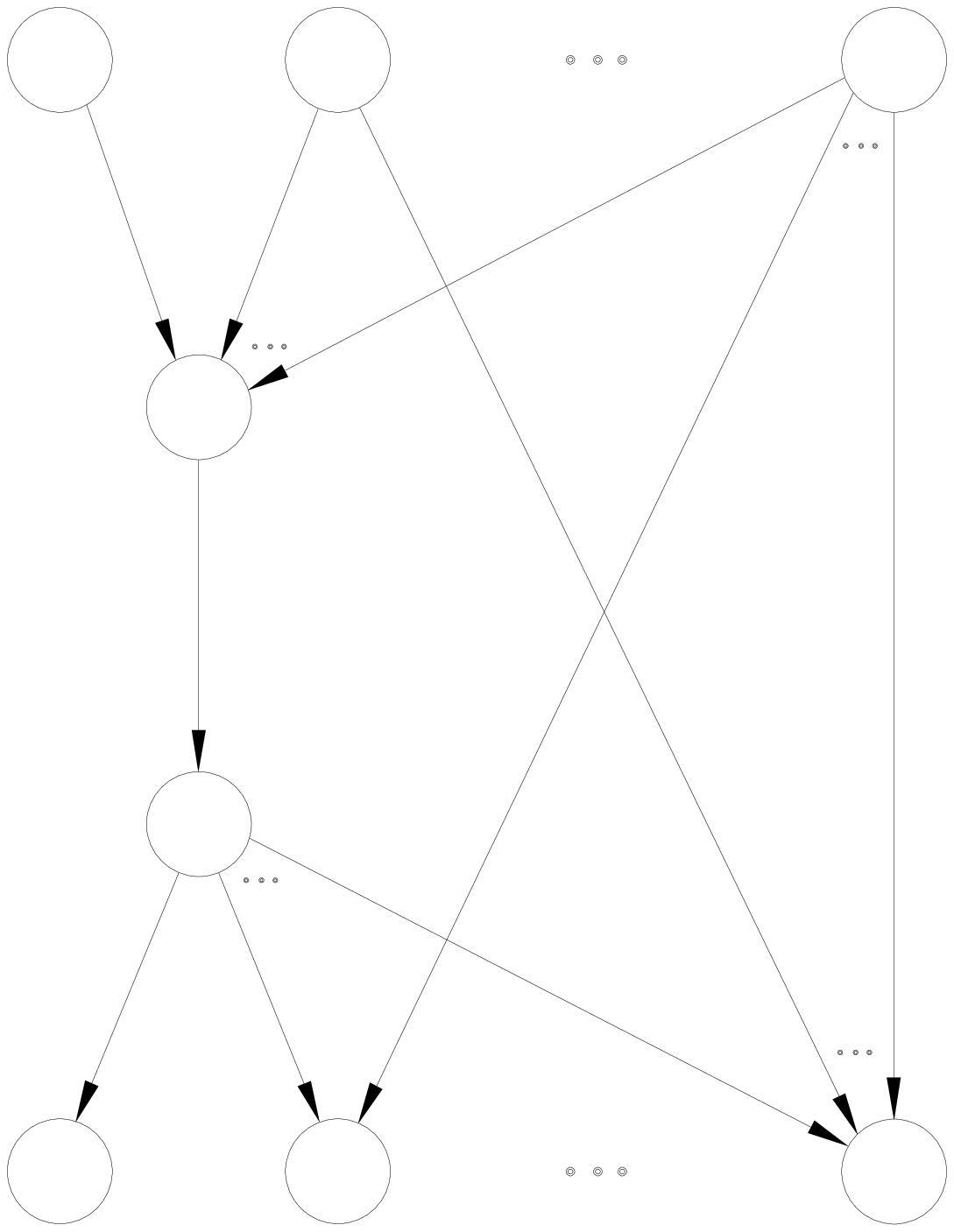} }
		\put(13.7,11.2){$s(1)$} \put(16.1,11.2){$s(2)$} \put(20,11.2){$s(k)$} \put(14.3,7.2){$u$} \put(14.3,3.5){$v$}
		\put(13.7,-.75){$t(k)$} \put(15.8,-.75){$t(k-1)$} \put(20,-.75){$t(1)$} \put(17,-1.8){(a)}
		
		\end{picture}
		\end{small}
		\caption{Network ${\cal N}_{1}$: (a)-The network for any $k$, (b)- An instance of ${\cal N}_{1}$ with $k=2$, (c)- An instance of ${\cal N}_{1}$ with $k=3$}
	\label{fig:Meagerness}
\end{figure}


\vspace{.5cm}
\begin{Lemma}
The value of the most meager cut in ${\cal N}_{1}$ is $1$.
\label{Lemma:meager}
\end{Lemma}
\begin{proof}
For any set of commodities ${\cal J} \subseteq {\cal I}$ we determine the meagerness of the most meager set of edges that isolates $\cal J$. First note that $\forall i \in {\cal I}$ there exists a path from $s(i)$ to $t(i)$ passing through the edge $e = (u,v)$. Thus any isolating set $A$ must contain $e$. Now consider the following two cases:
\begin{itemize}
\item If $|{\cal J}| = 1$, then $A = \{e\}$ and ${\cal M}(A) = 1$. 

\item If $|{\cal J}| \geq 2$. Let ${\cal J} = \{ i_{1},i_{2},\ldots,i_{|{\cal J}|} \}$ where without loss of generality $1 \leq i_{1} < i_{2} < \ldots < i_{|{\cal J}|} \leq k$. Since a cut $A$ must isolate all  the commodities in ${\cal J}$, it must isolate $s( i_{|{\cal J}|} )$ from all sinks $t(i_{1}), t({i_{2}}), \ldots, t(i_{|{\cal J}|} )$. From the structure of ${\cal N}_{1}$, there exists an edge from $s(i_{|{\cal J}|})$ to every sink node $t({i})$, $\forall i \in {\cal J} \bs \{i_{|{\cal J}|}\}$. Let $F$ be the set of such edges, then $|F| = |{\cal J}| - 1$ and for any isolating set, $A$, we must have $\{e\} \cup F \subseteq A$. Therefore, the capacity of any isolating set $A$ is at least $|{\cal J}|$. Therefore, ${\cal M}(A) \geq 1$.

\end{itemize}
The lemma follows by noting that ${\cal M}_{{\cal N}_{1}} = \min_{A \subseteq E} {\cal M}(A) = 1$.  
\end{proof}

\begin{Theorem}
There exist unit capacity, directed acyclic $k$-pairs networks where the network coding rate is $\Theta(|V|)$ multiplicative factor smaller than meagerness. 
\label{Theorem:meager}
\end{Theorem}

\begin{proof}
Consider the network ${\cal N}_{1}$ with $k$ sources as in Fig.\ref{fig:Meagerness}. Note that $t(k)$ must recover the message of $s(k)$, i.e. $X_{k}$, from the information carried by $e$. Similarly $t(k-1)$ recovers $X_{k-1}$ as a function of $X_{k}$ and $X_{e}$ and so on until $t(1)$ where $X_{1}$ is computed as a function of $X_{2}, \ldots, X_{k}$ and $X_{e}$. Thus we have
\begin{eqnarray}
&\hspace{-.4cm}t(k)\hspace{-.1cm} \mbox{\ \ gives:}& \hspace{-.2cm} H(X_{k},X_{e}) \leq H(X_{e}) \\
&\hspace{-.4cm}t(k-1)\hspace{.2cm} :& \hspace{-.2cm} H(X_{k-1}, X_{k}, X_{e}) \leq H(X_{k}, X_{e})\\
& \vdots & \nonumber \\
&\hspace{-1.4cm}t(2)\hspace{-.1cm} :& \hspace{-1.3cm} H(X_{2}, X_{3}, \ldots, X_{k}, X_{e}) \! \leq \! H(X_{3}, X_{4}, \ldots, X_{k}, X_{e}) \\
&\hspace{-1.4cm}t(1)\hspace{-.1cm} :& \hspace{-1.3cm} H(X_{1}, X_{2}, \ldots, X_{k}, X_{e}) \! \leq \! H(X_{2}, X_{3}, \ldots, X_{k}, X_{e})  
\end{eqnarray}
Applying forward substitution on the previous set of inequalities we obtain
\begin{eqnarray}
H(X_{e}) &\geq& H(X_{1}, X_{2}, \ldots, X_{k}, X_{e}) \\
         &\geq& H(X_{1}, X_{2}, \ldots, X_{k}) \label{eq2:meager_1} \\
         & = &  \sum_{i \in {\cal I}} H(X_{i}) \label{eq2:meager_2} \\
         &\geq& \sum_{i \in {\cal I}} r_{i} \label{eq2:meager_3}	\\
         & = & rk \label{eq2:meager_4}
\end{eqnarray}
where (\ref{eq2:meager_1}) follows since entropy is non-decreasing and (\ref{eq2:meager_2}) is due to the independence of sources. (\ref{eq2:meager_3}) and (\ref{eq2:meager_4}) follows from the definitions of rate and symmetric rate. Since edge $e$ has unit capacity we have $H(X_{e}) \leq 1$. Thus, the network coding rate is upper bounded as 
\begin{eqnarray}
r \leq \frac{1}{k}
\label{eq:rate_N1}
\end{eqnarray}
The theorem follows from (\ref{eq:rate_N1}) and Lemma \ref{Lemma:meager} by noting that $k = \Theta(|V|)$ for ${\cal N}_{1}$.
\end{proof}

\section{Undirected Networks}

Undirected $k$-pairs networks where considered in \cite{Li2004:2} where it was conjectured that network coding can not provide any rate improvement over routing. Since sparsity, ${\mathscr S}$, is an upper bound of both routing and network coding rates, the conjecture trivially holds true if the routing rate is equal to ${\mathscr S}$. Hence, to verify the validity of the conjecture, one must consider networks whose routing rate is strictly less than their sparsity. Hereafter we refer to such networks as \emph{gaped} networks. One such network that has been extensively considered is the Okamura-Seymour, OS, network \cite{Okamura1981:1}. The OS network is a $4$-pairs undirected network with $|V| = 5$, $|E| = 6$ whose Weiner bound $W = 3/4$, sparsity ${\mathscr S} = 1$ and routing rate equals $W$. 
It is not hard to verify that the underlying graph, $G$, of the OS network exhibits the smallest number of vertices among all underlying graphs, non-isomorphic to $G$, of $4$-pairs gaped networks (note that different networks might have the same unlabeled graph as their underlying graph). In \cite{Yeung2005:1} \cite{Harvey2006:1} \cite{Savari2006:1} it was independently shown that the network coding rate of the OS network is indeed equal to the routing rate. Hence, moving one step toward the $k$-pairs conjecture. 

Another class of networks for which the conjecture has been verified is the set of \emph{special bipartite} networks \cite{Harvey2006:1}. A summary of networks for which the conjecture holds true (including the ones obtained in the next two subsections) is listed below. Note that the classes in the list are not disjoint and might greatly intersect.

\begin{itemize}
\item $k$-pairs networks whose maximum achievable rates are equal to their sparsity. An undirect $k$-pairs network ${\cal N}$ is known to belong to this class of networks if
\begin{itemize}
\item ${\cal N}$ has one commodity \cite{Ford1956:1}.
\item ${\cal N}$ has two commodities, i.e. $k = 2$, \cite{Hu1963:1}.
\item ${\cal N}$ has an underlying planar graph, $G$, that can be drawn such that all source and sink nodes lay on the outer face of $G$ \cite{Okamura1981:1}.

\end{itemize}

\item ${\cal N}$ is the Okamura-Seymour network \cite{Yeung2005:1} \cite{Harvey2006:1} \cite{Savari2006:1}.

\item ${\cal N}$ is a special bipartite network \cite{Harvey2006:1}, \cite{Matula1990:1}.

\item ${\cal N}$ is the three commodity network in figure \ref{fig:Hu}.

\item ${\cal N}$ is a Type-I bipartite network, corollary \ref{cor:Type-I}.

\item ${\cal N}$ is a Type-II bipartite network, corollary \ref{cor:Type-II}.

\end{itemize}

\subsection{A Three-Commodity Network}
\label{Hu}

In this subsection we consider a three commodity network ${\cal N}_{2}$, Fig.\ref{fig:Hu}. It is known that routing can not achieve the sparsity (min-cut) of this network \cite{Hu1963:1}. To see this, consider all possible cuts in the network. It can be seen that sparsity is $4/3$. But the Wiener bound asserts that the routing rate can not exceed $8/7$. It is easy to advise a routing scheme achieving rate $8/7$ for ${\cal N}_{2}$. In the following we show that network coding does not have any rate advantages over routing and thus confirm the $k$-pairs conjecture over this network. Network ${\cal N}_{2}$ was considered in \cite{Savari2006:1} when all edges have capacity 2. The authors used an algorithm called \emph{progressive $d$-separating edge-set} or PdE to show that network coding can not achieve the rate tuple $(r_{a},r_{b},r_{g}) = (1,4,2)$ in ${\cal N}_{2}$. 
\begin{figure}[htbp]
		\centering \setlength{\unitlength}{0.38cm}
		\begin{picture}(18,13)(1,0)		
			\put(0,1){\includegraphics[width = 6.5cm,height = 5cm]{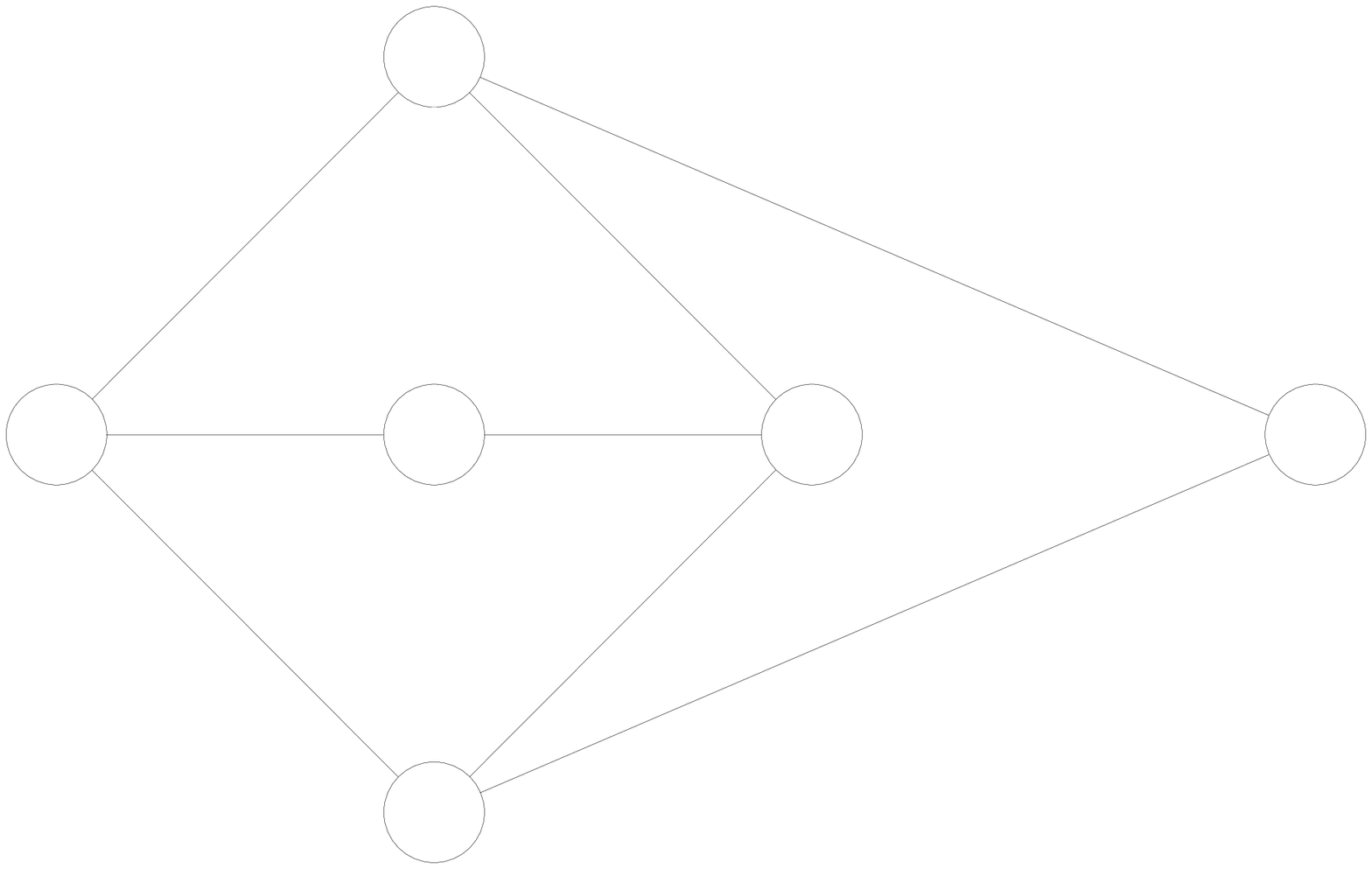} }
			\put(.5,6.2){$g$} \put(0,7.4){$s(g)$}
			\put(5.2,11){$a$} \put(4.7,12){$s(a)$} 
			\put(5.2,6.2){$b$} \put(5,7.4){$s(b)$} 
			\put(5.2,1.5){$c$} \put(4.7,0.5){$t(a)$}
			\put(9.8,6.2){$h$} \put(9.3,7.4){$t(g)$}
			\put(16,6.2){$f$} \put(15.5,7.2){$t(b)$}
		\end{picture}
		\caption{Network ${\cal N}_{2}$}
	\label{fig:Hu}
\end{figure}

\begin{Theorem}
The network coding rate for ${\cal N}_{2}$ is ${8}/{7}$.
\label{Theorem:N_Hu}
\end{Theorem}

\begin{proof}
Applying the input-output inequality at node $g$ we obtain
\begin{eqnarray}
H(\!X_{g},\!X_{ag},\!X_{bg},\!X_{cg},\!X_{ga},\!X_{gb},\!X_{gc}\!) \! \leq \! H(\!X_{g},\!X_{ag},\!X_{bg},\!X_{cg}\!)
\label{eq:d}
\end{eqnarray}
and at node $h$ we obtain
\begin{eqnarray}
H(X_{g},X_{ah},X_{bh},X_{ch},X_{ha},X_{hb},X_{hc}) \leq H(X_{ah},X_{bh},X_{ch})
\label{eq:e}
\end{eqnarray}
adding (\ref{eq:d}) and (\ref{eq:e}) and using submodularity in the LHS and the union bound in the RHS we obtain
\begin{eqnarray}
H(X_{g},X_{E^{\prime}}) \leq H(X_{ag},X_{bg},X_{cg}) + H(X_{ah},X_{bh},X_{ch})
\label{eq:de}
\end{eqnarray}
where $X_{E^{\prime}} = X_{E_{d}} \bs \{X_{af},X_{fa},X_{cf},X_{fc}\}$ and $E_{d}$ is the set of directed edges obtained from $E$. Now consider node $b$ and note that $X_{b}$ must be delivered through the directed edges $bh$ and $bg$, i.e., $X_{b}$ is a function of $X_{bg}$ and $X_{bh}$. Since $X_{bg}, X_{bh} \in X_{E^{\prime}}$, then from (\ref{eq:de}) we can write
%
\begin{eqnarray}
H(X_{b},X_{g},X_{E^{\prime}}) \leq  H(X_{ag},X_{bg},X_{cg}) + H(X_{ah},X_{bh},X_{ch})
\label{eq:deb}
\end{eqnarray}
The input-output inequality at $f$ gives $H(X_{b},X_{af},X_{cf},X_{fa},X_{fc}) \leq H(X_{af},X_{cf})$. Adding this to (\ref{eq:deb}) and using submodularity at the LHS we obtain $H(X_{b}) + H(X_{b},X_{g},X_{E}) \leq H(X_{af},X_{cf}) + H(X_{ag},X_{bg},X_{cg}) + H(X_{ah},X_{bh},X_{ch})$. From node $c$, we know $X_{a}$ must be recovered from $X_{gc},X_{hc},X_{fc} \in X_{E}$ and thus we can write
%
\begin{eqnarray}
H(X_{b})  +  H(X_{a},X_{b},X_{g},X_{E})   &\leq&   H(X_{ag},X_{bg},X_{cg}) 
 + H(X_{ah},X_{bh},X_{ch}) 
 + H(X_{af},X_{cf}) \\
&\leq&  H(X_{ag})\!+\!H(X_{bg})+ H(X_{cg})  
 + H(X_{ah}) \nonumber \\
&{}& + H(X_{bh})\!+\!H(X_{ch}) + H(X_{af})\!+\!H(X_{cf})
\label{eq:debc}
\end{eqnarray}
Since the sources are independent and $H(X_{a},X_{b},X_{g}) \leq H(X_{a},X_{b},X_{g},X_{E})$, (\ref{eq:debc}) gives
\begin{eqnarray}
H(X_a) + 2H(X_{b}) + H(X_{g})  &\leq&  H(X_{ag})\!+\!H(X_{bg})\!+\!H(X_{cg}) + H(X_{ah}) \nonumber \\
&{}& + H(X_{bh}) + H(X_{ch}) + H(X_{af}) + H(X_{cf})
\label{eq:debc2}
\end{eqnarray}

Now we apply the input-output inequality at node $a$ and write
\begin{eqnarray}
H(X_{a},X_{ga},X_{ha},X_{fa},X_{ag},X_{ah},X_{af}) \leq H(X_{a},X_{ga},X_{ha},X_{fa})
\label{eq:a}
\end{eqnarray}
and at node $c$ we obtain
\begin{eqnarray}
H(X_{a},X_{gc},X_{hc},X_{fc},X_{cg},X_{ch},X_{cf}) \leq H(X_{gc},X_{hc},X_{fc})
\label{eq:c}
\end{eqnarray}
computing (\ref{eq:a}) + (\ref{eq:c}) and using submodularity we get
\begin{eqnarray}
H(X_{a},X_{E^{\prime \prime}}) \leq H(X_{ga},X_{ha},X_{fa}) +  H(X_{gc},X_{hc},X_{fc})
\label{eq:ac}
\end{eqnarray}
where $X_{E^{\prime \prime}} = X_{E_{d}} \bs \{X_{gb},X_{bg},X_{bh},X_{hb}\}$. Since node $f$ demands $X_{b}$, $X_{b}$ must be a function of $X_{af},X_{cf} \in X_{E^{\prime \prime}}$. Therefore,
\begin{eqnarray}
H(X_{a},\!X_{b},\!X_{E^{\prime \prime}}) \! \leq \! H(X_{ga},\!X_{ha},\!X_{fa}) \! + \! H(X_{gc},\!X_{hc},\!X_{fc})
\label{eq:acf}
\end{eqnarray}
Applying input-output inequality at node $b$ gives $H(X_{b},X_{gb},X_{hb},X_{bg},X_{bh}) \leq H(X_{b},X_{gb},X_{hb})$. Adding this to (\ref{eq:acf}) and using submodularity we get $H(X_{a},X_{b},X_{E}) \leq H(X_{ga},X_{ha},X_{fa}) +  H(X_{gc},X_{hc},X_{fc}) + H(X_{gb},X_{hb})$. Since $X_{g}$ is recoverable from $X_{bh}, X_{ah}, X_{ch} \in X_{E}$, we can write
\begin{eqnarray}
H(X_{a},X_{b},X_{g},X_{E}) \leq H(X_{ga},X_{ha},X_{fa}) \!+\!  H(X_{gc},X_{hc},X_{fc})  + \ H(X_{gb},X_{hb})
\label{eq:acfe}
\end{eqnarray}
Thus, using independence of sources on the LHS and the union bound on the RHS we obtain
\begin{eqnarray}
H(X_{a})+H(X_{b})+H(X_{g}) 
\!\!\!\! &\leq& \!\!\!\! H(X_{ga}) \!+\! H(X_{ha}) \!+\! H(X_{fa})  \nonumber \\
&{}&   \!\!\!\!\!\! + H(X_{gc}) \!+\! H(X_{hc}) \!+\! H(X_{fc})  \nonumber \\
&{}& \!\!\!\!\!\! +   H(X_{gb}) \!+\! H(X_{hb})
\label{eq:acfe2}
\end{eqnarray}

Adding (\ref{eq:debc2}) and (\ref{eq:acfe2}) and noting that $H(X_{ij}) + H(X_{ji}) \leq 1$ we obtain
\begin{eqnarray}
2r_{a} \!+\! 3r_{b} \!+\! 2r_{g} \! \leq \! 2H(X_{a}) \!+\! 3H(X_{b}) \!+\! 2H(X_{g}) \! \leq \! 8 
\end{eqnarray}
The theorem follows by setting $r_{a} = r_{b} = r_{g} = r$ and noting that there exist a fractional routing scheme achieving rate $8/7$. 
\end{proof}

\subsection{Networks on Bipartite Graphs}

In this subsection we consider undirected $k$-pairs networks with underlying bipartite graphs. Let ${\cal N}$ be a $k$-pairs undirected network with a set of commodities ${\cal I}$ and an underlying bipartite graph $G(V \cup W,E)$. This problem was considered in \cite{Harvey2006:1} for the case when each commodity ${i \in \cal I}$ is such that $s(i)$ and $t(i)$ are located in the same partition $V$ or $W$. 
In this section we extend this study to any bipartite $k$-pairs network. Let ${\cal I}_{VV} = \{i \in S(V): t(i) \in V\}$ be the set of all commodities whose sources and sinks are in $V$, also let ${\cal I}_{VW} = \{i \in S(V): t(i) \in W\}$ be the set of all commodities whose sources are in $V$ and sinks are in $W$. On the other hand, let ${\cal I}_{WW} = \{i \in S(W): t(i) \in W\}$ be the set of all commodities whose sources and sinks are in $W$ and ${\cal I}_{WV} = \{i \in S(W): t(i) \in V\}$ be the set of commodities from $W$ to $V$.

The following is a Lemma required in proving the next theorem. We present the lemma without a proof and refer the interested reader to \cite{Harvey2006:1} where a stronger result was proven.

\begin{Lemma}
For any collection of sets $A_{1},\ldots, A_{n}$
\[
\sum_{i=1}^{n} H(X_{A_{i}}) \geq H(X_{\bigcup_{i=1}^{n}A_{i}}) + H(X_{\bigcup_{1\leq i<j\leq n}^{n} {A_{i} \cap A_{j}}}) 
\]
\label{Lemma:BiP}
\end{Lemma}

\begin{Theorem}
For an undirected $k$-pairs network on a bipartite graph $G(V \cup W,E)$ with a set of commodities ${\cal I} = {\cal I}_{VV} \cup {\cal I}_{VW} \cup {\cal I}_{WV} \cup {\cal I}_{WW}$, the network coding rate is bounded as
\[
r \leq \frac{|E|}{(|{\cal I}_{VW}| + |{\cal I}_{WV}|) + 2(|{\cal I}_{VV}| + |{\cal I}_{WW}|)}
\]
\label{Theorem:BiP}
\end{Theorem}

\begin{proof}
Applying the input-output inequality at each node $v \in V$ we can write
\begin{eqnarray}
\!\!\!\!\!\!\!\!\!\! H(X_{S(v)},\! X_{\mbox{In}(v)}, \! X_{\mbox{Out}(v)}, \! X_{T(v)}) \!\!\!\!\! &\leq& \!\!\!\!\! H(X_{S(v)}, X_{\mbox{In}(v)}) \\
&{}&  \!\!\!\!\!\!\!\!\!\!\!\!\!\!\!\!\!\!\!  \leq H(X_{S(v)}) \!+\! H(X_{\mbox{In}(v)})
\label{eq:BiP1}
\end{eqnarray}
adding (\ref{eq:BiP1}) over all $v \in V$ we can write
\begin{eqnarray}
H(X_{S(V)}, &{}& \hspace{-.9cm} X_{\mbox{In}(V)}, X_{\mbox{Out}(V)}, X_{T(V)}) + H(X_{{\cal I}_{VV}}) \nonumber \\ 
&\leq& \sum_{v \in V} H(X_{S(v)}) + \sum_{v \in V} H( X_{\mbox{In}(v)} ) \label{eq:BiP2_1} \\
&\leq& \sum_{v \in V} H(X_{S(v)}) + \sum_{v \in V} \sum_{e \in \mbox{In}(v)} H(X_{e}) \label{eq:BiP2_2} \\
&=& \sum_{v \in V} H(X_{S(v)}) + \sum_{e \in \mbox{In}(V)} H(X_{e}) \label{eq:BiP2_3}
\end{eqnarray}
The LHS of (\ref{eq:BiP2_1}) follows from Lemma \ref{Lemma:BiP} where the second term is obtained by noting that $S(v) \cap T(u) \subseteq {\cal I}_{VV}$, $\forall u,v \in V$ and $\bigcup_{\{u,v\} \subseteq V}{S(u) \cap T(v)} = {\cal I}_{VV}$. (\ref{eq:BiP2_2}) follows from the union bound and (\ref{eq:BiP2_3}) follows since a node $v \in V$ has no neighbors in $V$ (from the definition of a bipartite). Since the underlying graph is bipartite we know that $\mbox{In}(V) \cup \mbox{Out}(V) = E_{d}$. Using this in the LHS of (\ref{eq:BiP2_3}) and noting that all commodities must be communicated via $E_{d}$, we can write
\begin{eqnarray}
H(X_{\cal I}) + H(X_{{\cal I}_{VV}}) \leq \sum_{v \in V} H(X_{S(v)}) + \sum_{e \in \mbox{In}(V)} H(X_{e})
\label{eq:BiP_V}
\end{eqnarray}
Using the same argument at the partition $W$, we obtain
\begin{eqnarray}
H(X_{\cal I}) + H(X_{{\cal I}_{WW}}) \leq \sum_{w \in W} H(X_{S(w)}) + \sum_{e \in \mbox{In}(W)} H(X_{e})
\label{eq:BiP_W}
\end{eqnarray}
Adding (\ref{eq:BiP_V}) and (\ref{eq:BiP_W}) and noting that $\mbox{In}(W) = \mbox{Out}(V)$ we obtain
\begin{eqnarray}
\!\!\!\!\!\!\!\!\!\!\!\!\! 2H(X_{\cal I}) \! + \! H(X_{{\cal I}_{VV}}) \! + \! H(X_{{\cal I}_{WW}}) 
\!\!\!\! &\leq& \!\!\!\!\! \!\!\! \sum_{v \in V \cup W} \!\!\!\!\! H(X_{S(v)}) \! + \! |E| \label{eq:fin} \\
&=& \sum_{i \in {\cal I}} H(X_{i}) + |E|
\end{eqnarray}
where in (\ref{eq:fin}) we used the fact that $\sum_{e \in \mbox{In}(W)} H(X_{e}) + \sum_{e \in \mbox{Out}(W)} H(X_{e}) \leq \sum_{e \in E} 1 = |E|$. Using the independence of the sources on the LHS to get
\begin{eqnarray}
\sum_{i \in {\cal I}} H(X_{i}) + \sum_{i \in {\cal I}_{VV}} H(X_{i}) + \sum_{i \in {\cal I}_{WW}} H(X_{i}) \leq |E|
\end{eqnarray}
Since $r_{i} \leq H(X_{i})$, the previous result become
\begin{eqnarray}
\sum_{i \in {\cal I}} r_{i} + \sum_{i \in {\cal I}_{VV}} r_{i} + \sum_{i \in {\cal I}_{WW}} r_{i} \leq |E|
\end{eqnarray}
and the symmetric rate can be bounded as
\begin{eqnarray}
r \leq \frac{|E|}{|{\cal I}| + |{\cal I}_{VV}| + |{\cal I}_{WW}|}
\end{eqnarray}
The theorem follows by noting that $|{\cal I}| = |{\cal I}_{VV}| + |{\cal I}_{VW}| + |{\cal I}_{WV}| + |{\cal I}_{WW}|$.
\end{proof}

In the following we provide two subclasses of $k$-pairs bipartite networks for which the $k$-pairs conjecture holds true. Let ${\cal N}$ be a $k$-pairs bipartite network with a set of commodities ${\cal I}$ and an underlying \emph{complete} bipartite graph $G(V \cup W, E)$, i.e, $G = K_{|V|,|W|}$. Then ${\cal N}$ is a Type-I $k$-pairs bipartite network if the following three conditions are satisfied
\begin{itemize}
\item For every unordered pair of distinct vertices $u, v \in V$ there exists a commodity $i \in {\cal I}$ such that $\{s(i),t(i)\} = \{u,v\}$.
\item For every unordered pair of distinct vertices $u, v \in W$ there exists a commodity $i \in {\cal I}$ such that $\{s(i),u(i)\} = \{u,v\}$.
\item There is no commodity $i \in {\cal I}$ such that $s(i) \in V$ and $t(i) \in W$ or $s(i) \in W$ and $t(i) \in V$.
\end{itemize} 
On the other hand, ${\cal N}$ is a Type-II network if for every unordered pair of distinct vertices $u, v \in V \cup W$ there exists a commodity $i \in {\cal I}$ such that $\{s(i),u(i)\} = \{u,v\}$.

\vspace{.4cm}
\begin{Corollary}
The $k$-pairs conjecture holds for Type-I networks.
\label{cor:Type-I}
\end{Corollary}

\begin{proof}(sketch)
The corollary follows by showing that there exists a fractional routing scheme achieving the rate $r$ in theorem \ref{Theorem:BiP} with $|{\cal I}_{VW}| = |{\cal I}_{WV}| = 0$, $|{\cal I}_{VV}| = \binom{|V|}{2}, |{\cal I}_{WW}| = \binom{|W|}{2}$ and $|E| = |V|.|W|$. The proof of the existence of such routing scheme is omitted.
\end{proof}

\begin{Corollary}
The $k$-pairs conjecture holds for Type-II networks.
\label{cor:Type-II}
\end{Corollary}

\begin{proof}(sketch)
Similar to the proof of the previous corollary where in this case we have $|{\cal I}_{VW}| + |{\cal I}_{WV}| = |V|.|W|$, $|{\cal I}_{VV}| = \binom{|V|}{2}$ and $|{\cal I}_{WW}| = \binom{|W|}{2}$.
\end{proof}

\section{Acknowledgment}

The authors wish to thank one of the anonymous reviewers for pointing out that a similar result to theorem \ref{Theorem:meager} was previously obtained by N. Harvey and R. Kleinberg \cite{Harvey2005:1}.

\bibliography{C:/Research_Latex/bibliographys/Network_Coding}

\end{document}